\documentclass[11pt]{article}
\usepackage{a4wide}
\usepackage{amsmath,amssymb,amsthm}
\usepackage{graphics,graphicx}
\usepackage{eucal} 
\usepackage{microtype}
\usepackage{hyperref}

\usepackage[ruled]{algorithm2e}

\SetAlFnt{\small}
\SetAlCapFnt{\small}
\SetAlCapNameFnt{\small}
\SetAlCapHSkip{0pt}
\IncMargin{-\parindent}

\newcommand{\F}{\mathcal{F}}
\newcommand{\ignore}[1]{}

\newtheorem{theorem}{Theorem}[]
\newtheorem{lemma}[theorem]{Lemma}

\newcommand{\field}[1]{{\textrm{GF}}({#1})}
\newcommand{\headingnsp}[1]{\noindent{\bf #1.\ }}%
\newcommand{\heading}[1]{\medskip\noindent{\bf #1.\ }}%

\begin{document}

\title{Fast Witness Extraction Using a Decision Oracle\thanks{A preliminary conference abstract of this work has appeared as A. Bj\"orklund, P. Kaski, and \L. Kowalik, ``Fast Witness Extraction Using a Decision Oracle'', Proceedings of the 22nd Annual European Symposium on Algorithms (ESA 2014, Wroc{\l}aw, September 8--10, 2014),  Lecture Notes in Computer Science vol. 8737, Springer, 2014, pp.~149--160.}}
\author{
  Andreas Bj\"orklund\thanks{Department of Computer Science, Lund University, Sweden; email: \texttt{andreas.bjorklund@yahoo.se}}
  \and
  Petteri Kaski\thanks{Helsinki Institute for Information Technology HIIT, Department of Information and Computer Science, Aalto University, Finland; email:
\texttt{petteri.kaski@aalto.fi}. Supported in part by the Academy of Finland, Grants 252083 and 256287}
  \and
  \L{}ukasz Kowalik\thanks{Institute of Informatics, University of Warsaw, Poland; email: \texttt{kowalik@mimuw.edu.pl}. Supported by National Science Centre of Poland (grant 2013/09/B/ST6/03136).}
}

\date{}

\maketitle

\begin{abstract}
The gist of many (NP-)hard combinatorial problems is to decide
whether a universe of $n$ elements contains a {\em witness} consisting 
of $k$ elements that match some prescribed pattern. 
For some of these problems there are known advanced algebra-based 
FPT algorithms which solve the decision problem but do not return the witness.
We investigate techniques for turning such a YES/NO-decision oracle into 
an algorithm for extracting a single witness,
with an objective to obtain practical scalability for large values of $n$.
By relying on techniques from combinatorial group testing, we demonstrate
that a witness may be extracted with $O(k\log n)$ queries to 
either a deterministic or a randomized set inclusion oracle with one-sided
probability of error. Furthermore, we demonstrate through implementation
and experiments that the algebra-based FPT algorithms are practical,
in particular in the setting of the $k$-path problem. Also discussed
are engineering issues such as optimizing finite field arithmetic.
\end{abstract}


\section{Introduction}

The gist of many (NP-)hard combinatorial problems is to {\em decide} 
whether a universe of $n$ elements contains a {\em witness} consisting 
of $k$ elements that match some prescribed pattern. In the positive case 
this is naturally followed by the task of {\em extracting} the elements 
of one such witness. 

As a result of advances in fixed-parameter tractability, 
many such hard problems are now known to admit algorithms that run
in linear (or low-order polynomial) time in the size of the universe
$n$, and where the complexity of the problem can be isolated to
the size of the witness $k$. That is, the running times obtained
are of the form $O(f(k)\cdot n)$ for some rapidly growing function 
$f(k)$ of $k$. This makes such algorithms ideal candidates for 
practical applications that must consider large inputs, 
that is, large values of $n$. For example, a recent randomized 
algorithm for the $k$-sized graph motif problem runs in time 
$O(2^kk^2(\log k)^2\cdot e)$, where $e$ is the number of edges in 
the input graph \cite{BjorklundKaskiKowalik2013}. 

Despite scalability to large inputs, some such advanced parameterized 
algorithms (like the ones for graph motif~\cite{BjorklundKaskiKowalik2013} or for $k$-path~\cite{BjorklundHusfeldtKaskiKoivisto2010}) have an inherent handicap from a concrete algorithm engineering
perspective. {\em They only solve the decision problem.}
In applications, however, one needs access to the witnesses,
which puts forth the question whether one can efficiently extract a
witness or list all witnesses, using the algorithm for the decision
problem as an {\em oracle} (black-box subroutine), and without losing 
the scalability to large inputs.

This paper studies the question of efficiently turning a decision oracle 
into an algorithm for witness extraction over the 
universe $U=\{1,2,\ldots,n\}$. Let $\mathcal{F}\subseteq 2^U$ be the (unknown) family of witnesses. We focus on the following oracle:
\begin{description}
\item[Inclusion oracle]
Given a query set $Y\subseteq U$, the oracle answers (either YES or NO) 
whether there exists at least one witness $W\in\mathcal{F}$ such that 
$W\subseteq Y$. We can motivate this type of oracle by observing
that most problems have natural self-reducibility that we can use 
to narrow down the universe from $U$ to $Y$ (e.g. take the subgraph 
induced by the set $Y$ of vertices) and then run the decision algorithm.
\end{description}

In the oracle setting there are at least two natural ways to measure
the efficiency of witness extraction. 
\begin{description}
\item[Number of oracle queries]
This measure has been extensively studied in the domain
of {\em combinatorial group testing} \cite{DuHwang2000}, 
where the canonical task is to identify $k$ {\em defective} items 
from a population of $n$ items, with the objective of minimizing
the number of tests%
\footnote{In the setting of classical group testing, a single test 
on a set of items determines whether the set contains at least 
one defective item.}
(oracle queries) required to identify all the
defectives. While this measure does not reflect accurately the amount 
of computing resources invested in our context---indeed, different 
oracle queries in general do not use the same amount of resources---the 
group testing perspective enables information-theoretic lower
bounds and supplies useful algorithmic techniques for extraction.
\item[Total running time]
Assuming we have bounds on the running time of the oracle
as a function of $n$ and $k$, we can bound the running time 
of extraction of witnesses by taking the sum of
the running times of the oracle queries. It turns
out that we get fair control over the total running
time already if we know that the running time of the oracle 
scales at least linearly in $n$. 
\end{description}

The objectives of this paper are threefold. 
(a) First, we draw from techniques in classical group testing to arrive 
at efficient witness extraction algorithms for inclusion oracles both 
in deterministic and in randomized settings with one-sided error. 
(b) Second, we show examples of parameterized problems which can be solved efficiently {\em in practice} by 
a combination of an FPT decision oracle and a group-testing algorithm;
in particular, for the $k$-path problem our experimental results show that one can find a $14$-vertex witness in a 2000-vertex graph within a minute on a typical laptop. 
(c) Third, we discuss some non-obvious choices we made during the implementation: namely the choice of the $\field{2^q}$ arithmetic implementation; we believe our findings might be useful for implementations of other algorithms applying $\field{2^q}$ arithmetic.

To set up a trivial baseline for performance comparisons, it is not 
difficult to see that $\Theta(n)$ queries to an inclusion oracle suffice 
to extract a witness---simply delete points from the universe one by one, 
with each deletion followed by an oracle query on the remaining points. 
If the oracle answers NO, we know the deleted point was essential and 
insert it back. When the process finishes the points that remain form 
a witness. This, however, is not particularly efficient since each oracle 
query costs at least $O(f(k)\cdot n)$ time, raising the total running time 
to $O(f(k)\cdot n^2)$ and making the approach impractical for 
large $n$.

\heading{\bf Our Results on Extraction}
We begin by transporting techniques from group testing \cite{DuHwang2000} to arrive at more efficient witness extraction. Our first contribution merely amounts to observing that the so-called {\em bisecting algorithm} \cite{DuHwang1993} can be translated to work with an inclusion oracle and in the presence of one or more witnesses. We also observe that taking into account the total running time of the algorithm, the baseline cost of a factor $O(n)$ in running time can be lowered to $O(k)$ if the running time of the oracle is at least linear in $n$, which is the case in most applications. These observations are summarized in Theorem~\ref{thm:main-deterministic}.

Let $\mathcal{F}$ be a nonempty family witnesses, each of size at most $k$, 
over an $n$-element universe, $n,k\geq 1$.
We say that a function $g:\mathbb{N}\rightarrow\mathbb{N}$ is 
{\em at least linear} if for all $n_1,n_2\in\mathbb{N}$ it holds
that $g(n_1)+g(n_2)\leq g(n_1+n_2)$.

\begin{theorem}[Deterministic Extraction]
\label{thm:main-deterministic}
There exists an algorithm that extracts a witness in $\mathcal{F}$ 
without knowledge of $k$ using at most 
\[
Q(n,k)=2k\biggl(\log_2\frac{n}{k}+2\biggr)
\]
queries to a deterministic inclusion oracle. Moreover, suppose the
oracle runs in time $T(n,k)=O(f(k)g(n))$ for a function $g$ that 
is at least linear. Then, there exists an algorithm that 
extracts a witness in $\mathcal{F}$ in time $O(k\cdot T(2n,k))=O(f(k)\cdot k\cdot g(2n))$. 
\end{theorem}

Currently the fastest known parameterized algorithms 
in many cases use randomization. Thus in practice one must be able to 
cope with decision oracles that may give erroneous answers, for example 
it is typically the case that the decision algorithm produces false 
negatives with at most some small probability, but false positives do 
not occur~\cite{BjorklundHusfeldtKaskiKoivisto2010,BjorklundKaskiKowalik2013,Williams2009,Koutis2012}. 

Let us assume that the probability of a false negative is $p\leq \frac{1}4$. 
Beyond the absence of false positives, a further observation to our 
advantage is that typically {\em witnesses may be checked}, deterministically, 
and essentially at no computational cost compared with the
execution of even one oracle query. That is, we have available a subroutine
that takes a candidate witness $W\subseteq U$ as input and returns 
whether $W\in\mathcal{F}$. We make this assumption in what follows.
Thus having access to a randomized inclusion oracle enables deterministic
extraction, but with randomized running time. These observations are 
summarized in Theorem~\ref{thm:main-randomized}.

\begin{theorem}[Las Vegas Extraction]
\label{thm:main-randomized}
There exists an algorithm that extracts a witness in $\mathcal{F}$ 
without knowledge of $k$ using in expectation at most 
$O(k\log n)$ queries to a randomized inclusion oracle that has no 
false positives but may output a false negative with probability at most 
$p\leq \frac{1}4$. Moreover, suppose the oracle runs in time $T(n,k)=O(f(k)g(n))$ 
for a function $g$ that is at least linear. Then, there exists an algorithm 
that extracts a witness in $\mathcal{F}$ in 
time $O(k\cdot T(2n,k)+(k\log k)\cdot T(2k,k))$. 
\end{theorem}

\headingnsp{An Application: $k$-Path}
The $k$-path problem is one of the basic NP-complete problems, a natural parameterized version of the Hamiltonian Path problem.
In this problem we are given an undirected connected graph $G = (V, E)$, and a natural number $k$. 
The goal is to find a simple path on $k$ vertices in $G$. 
Denote by $n=|V|$ and $m=|E|$.
In terms of dependence on $k$, the currently fastest algorithm is due to Bj\"orklund, Husfeldt, Kaski, and Koivisto~\cite{BjorklundHusfeldtKaskiKoivisto2010} and can be tuned to run in $1.66^kk^{O(1)}m$ time.
It uses algebraic tools and only solves the corresponding decision problem.
We applied a simplified version of this algorithm, slightly easier to implement, which runs in $O(2^kkm)$ time, assuming that finite field arithmetic operations take constant time (cf.~\cite{fpt-textbook}). 
The algorithm evaluates a certain polynomial of degree $d=2k-1$ over the finite field $\field{2^q}$, which turns out to be a generating function of all witnesses. The algorithm is randomized, and it may return a false negative. The failure probability is bounded by $\frac{2k-1}{2^q}$, hence by choosing $q$ large enough we can assume it is at most $\frac{1}4$, as required by Theorem~\ref{thm:main-randomized}.

Our universe $U$ is the set of edges of the input graph and we are extracting witnesses with exactly $k-1$ edges. 
By Theorem~\ref{thm:main-randomized} we obtain an algorithm with expected running time $O(2^kk^2\cdot m)$ for witness extraction.

However, when we consider actual {\em implementation} the above approach should be refined as follows.
First set the universe $U$ to be the set of {\em vertices} and find the set of $k$ vertices $S$ which contains a $k$-vertex path.
Next, set the universe $U$ to be the set of {\em edges} in the induced graph $G[S]$, and find the witness. 
By Theorem~\ref{thm:main-randomized}, for dense graphs this can give a factor two speed-up.

\heading{A Computational Biology Application: Graph Motif}
In the graph motif problem we are given an undirected connected graph $G = (V, E)$, a vertex coloring $c:V\rightarrow C$, and a multiset $M$ of cardinality $k$ consisting of colors in the set $C$. The goal is to find a subset $S\subseteq V$ such that the induced subgraph $G[S]$ is connected, and the {\em multiset} $c(S)$ of colors of the vertices of $S$ is equal to $M$. Note that $|S|=k$. 
This problem has important applications in querying patterns in protein-protein interaction (PPI) networks, see e.g.~\cite{BrucknerRECOMB09}.
Although problem is NP-hard, in the data instances coming from this application the graph size is of order of thousands and the pattern is very small (according to~\cite{BrucknerRECOMB09} the number of edges is 21275 for fly and 28972 for human, and $k\in\{4,\ldots,25\}$), i.e.\ they are perfectly suited for FPT algorithms.
A recent randomized decision algorithm \cite{BjorklundKaskiKowalik2013} solves the corresponding decision problem by evaluating a certain polynomial of degree $d=3k-1$ over the finite field $\field{2^q}$, which turns out to be a generation function of all witnesses.
Its running time is dominated by performing $O(2^kk^2\cdot m)$ arithmetic operations in $\field{2^q}$, where $m$ is the number 
of edges in the input graph. The algorithm returns false negatives with probability bounded by $\frac{d}{2^q}$, hence by choosing $q$ large enough we can assume it is at most $\frac{1}4$. It follows that we can use it as a randomized oracle in the algorithm described in Theorem~\ref{thm:main-randomized}.

Our universe $U$ is be the set of vertices of the input graph. By Theorem~\ref{thm:main-randomized} we obtain an algorithm with expected running time 
$O(2^kk^3\cdot m)$ for witness extraction, assuming that finite field arithmetic operations take constant time.

\heading{Further applications} 
The list of problems which a) fall into our witness extraction framework and b) have the property that asymptotically fastest decision algorithms do not return a witness includes Steiner tree~\cite{jesper-st}, $q$-set packing~\cite{BjorklundHusfeldtKaskiKoivisto2010}, $q$-dimensional matching~\cite{BjorklundHusfeldtKaskiKoivisto2010}, Steiner cycle (aka $K$-cycle)~\cite{BjorklundHusfeldtTaslaman2012,Wahlstrom2013}, directed rural postman problem~\cite{Wahlstrom+2013}.

\heading{Related and Previous Work}
The relations between the time complexity of decision problems and their search versions were studied by Fellows and Langston~\cite{FellowsL:STOC89}.

Independently of our work, Hassidim, Keller, Lewenstein, and Roditty~\cite{hklr_wads} presented a {\em randomized} algorithm that extracts a witness for the (weighted) $k$-path problem using $O(k\log n)$ calls to a decision oracle, {\em in expectation}. 
Their approach is to discard random subsets (of size $n/k$) of the vertex set as long as the resulting instance still contains the solution. The bisecting algorithm~\cite{DuHwang1993} that we extend in this paper can be seen as a cleaner version of this idea.
First, in the bisecting algorithm larger sets get discarded. Second, the bisecting algorithm is deterministic. 
Hassidim et al.\ do not analyze how the time of their algorithm is influenced by the fact that the oracle is randomized. From an asymptotic perspective this is not needed because one can repeat each oracle call multiple times to reduce the error probability below an arbitrary threshold. However, in practice this is an unnecessary (though only constant-factor) slow-down, which we seek to avoid in what follows.

\heading{Implementation and Experiments}
We implemented in C the $O(2^kkm)$-time decision algorithm for the $k$-path problem and the algorithm from Theorem~\ref{thm:main-randomized}, which we call `fifo' on the charts. 
The crucial part of implementation of the decision oracle is the finite field arithmetic. 
Somewhat unexpectedly, we found that to optimize the running time, a {\em different} method should be chosen depending on whether we use the oracle just once (e.g.\ check whether there is a witness) or whether it is used in combination with the algorithm from Theorem~\ref{thm:main-randomized} to find a witness.
Details can be found in Section~\ref{sec:gf-impl}.

\begin{figure}[t]
\begin{minipage}[b]{0.5\linewidth}
\centering
\includegraphics[width=\textwidth]{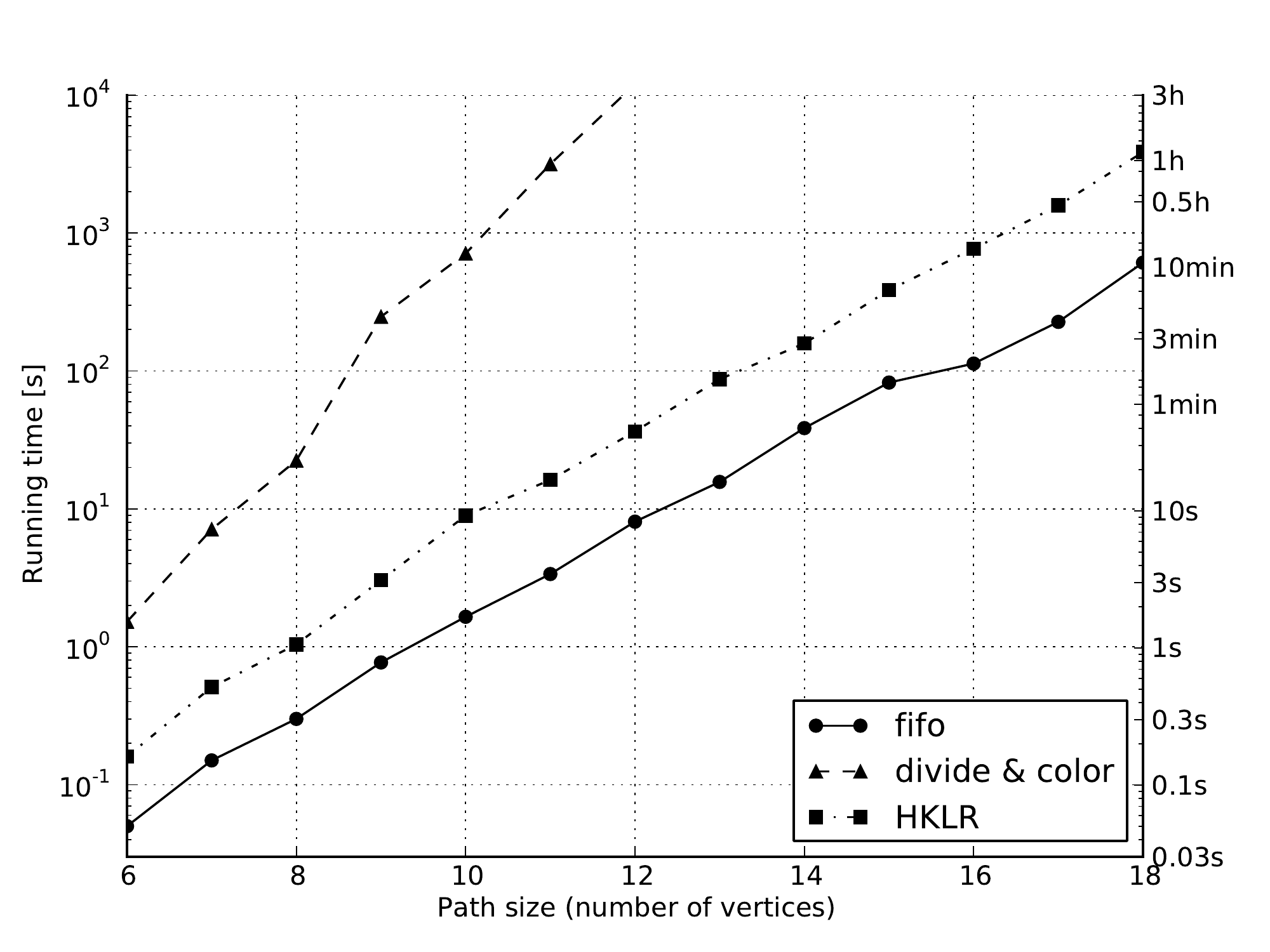}
\end{minipage}
\begin{minipage}[b]{0.5\linewidth}
\centering
\includegraphics[width=\textwidth]{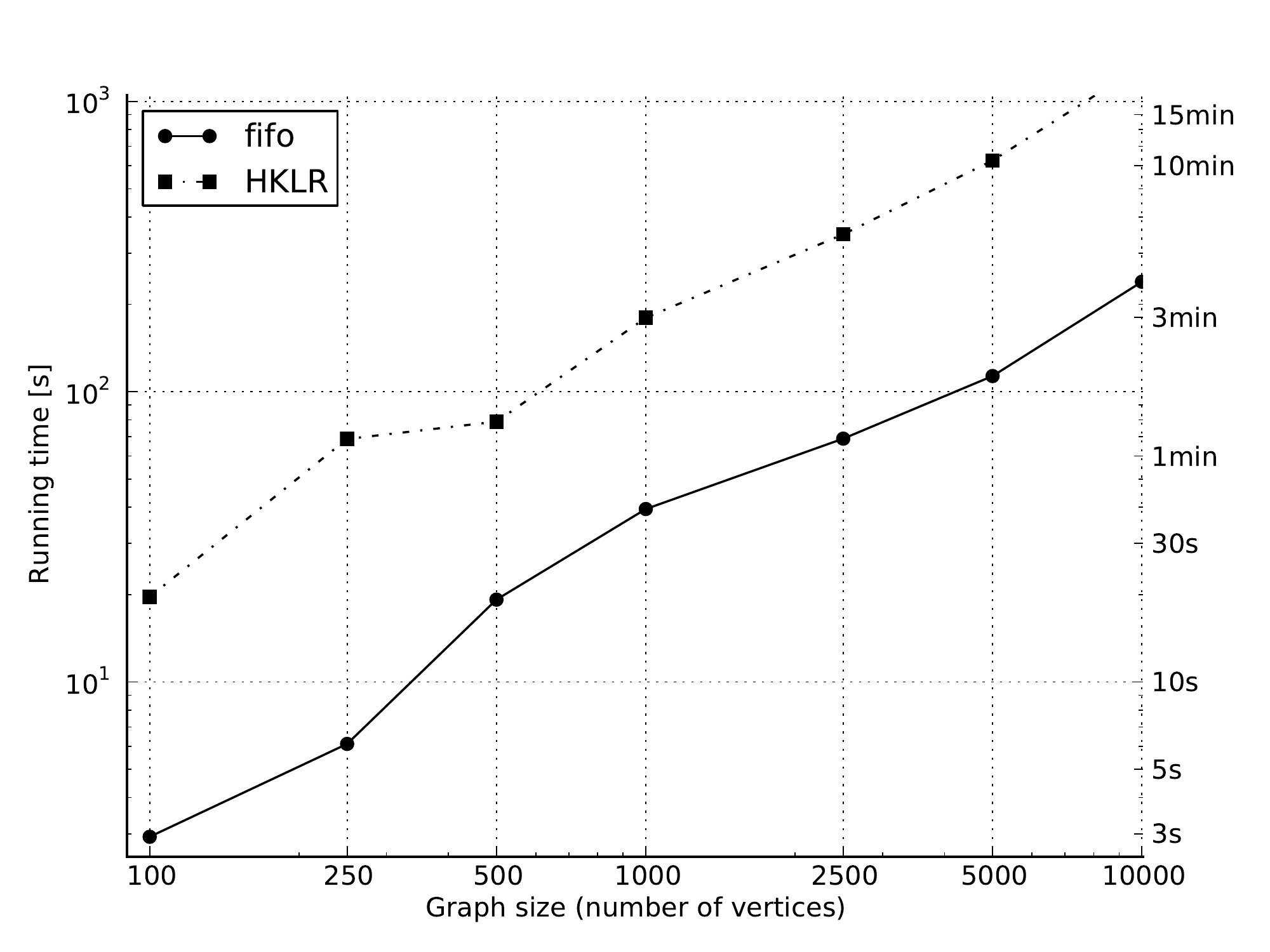}
\end{minipage}
\begin{minipage}[b]{0.5\linewidth}
\centering
\includegraphics[width=\textwidth]{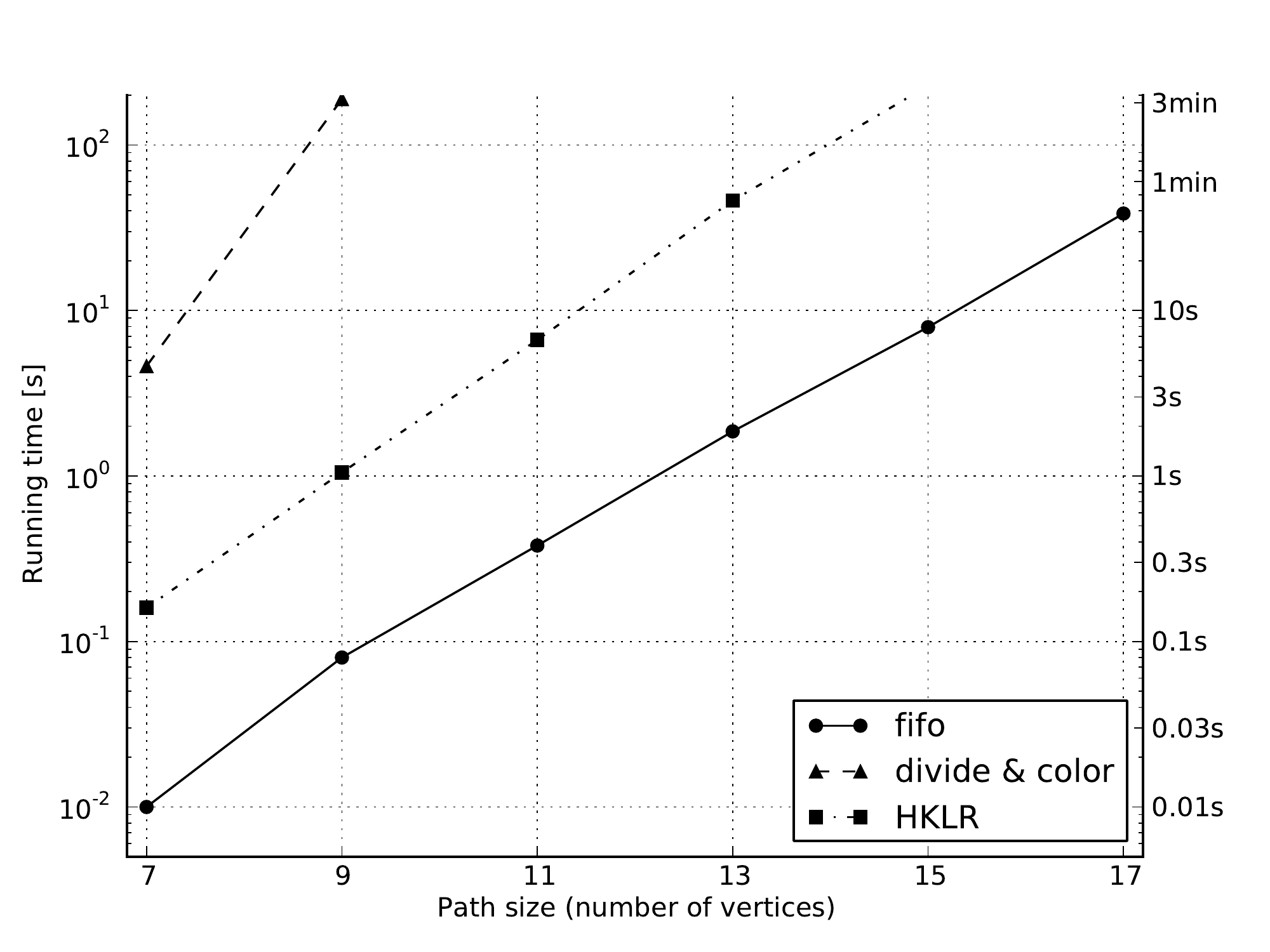}
\end{minipage}
\begin{minipage}[b]{0.5\linewidth}
\centering
\includegraphics[width=\textwidth]{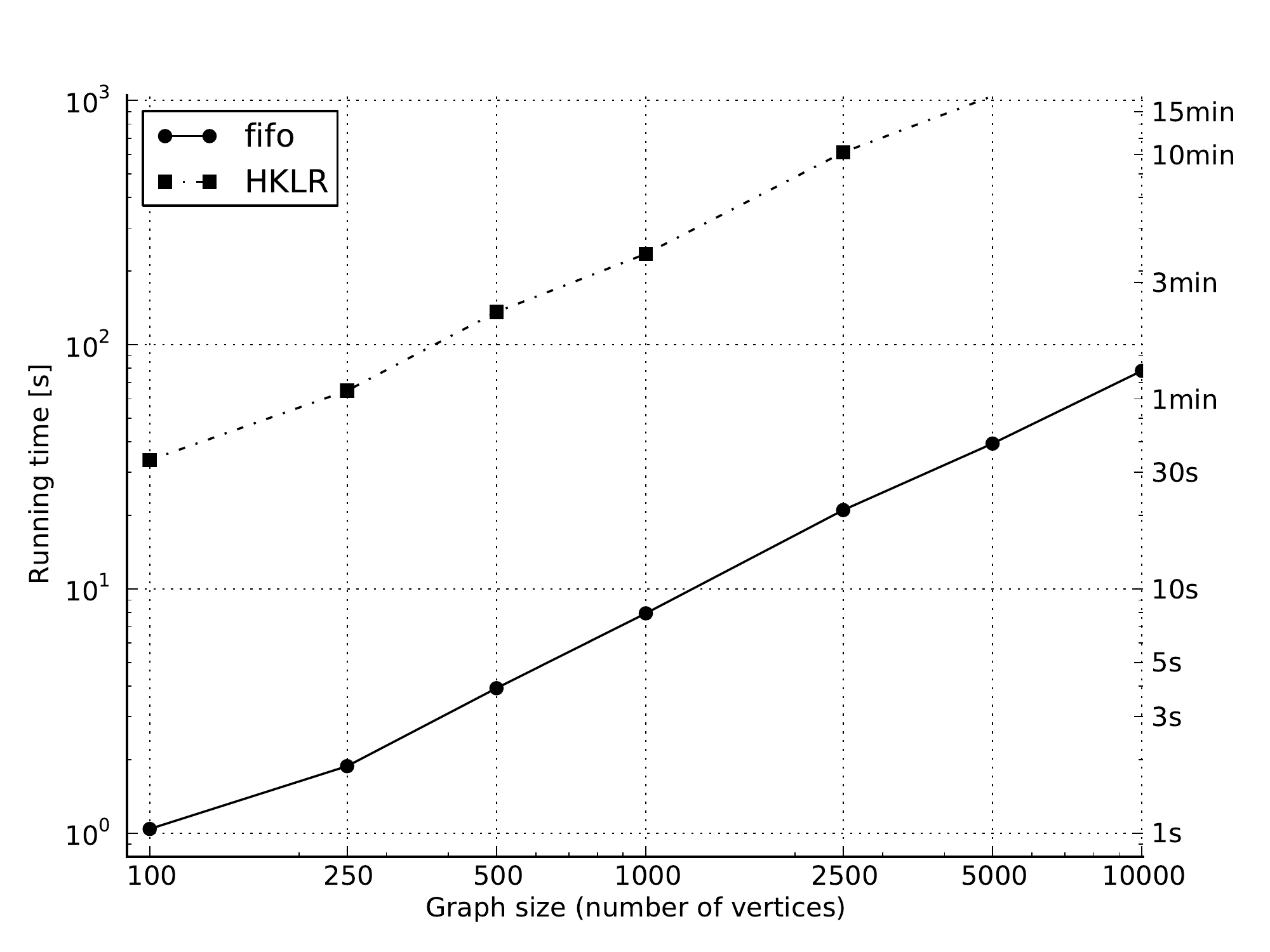}
\end{minipage}
\caption{\label{fig:alg_charts}
         Running times of various algorithms for a graph with exactly one witness (upper charts) and  $\Omega(n^2)$ witnesses (lower charts). 
         Each running time on the graph is the median of 5 runs for the same input instance.
         The left charts: a 1000-vertex graph and $k\in\{6,7,\ldots,18\}$.
         The right charts: $k=14$ (upper) or $k=15$ (lower) and the number of vertices varies. Running times on a 2.53-GHz Intel Xeon CPU.
         }
\end{figure}

We run a series of experiments on a single 2.53-GHz Intel Xeon CPU. 
We compare the fifo algorithm with two other natural candidates. The first is the witness extraction algorithm of Hassidim et al.~\cite{hklr_wads} combined with the $O(2^kkm)$-time inclusion oracle, called `HKLR' on the charts. The second is the $O(4^kk^{2.7}m)$-time algorithm of Chen et al.~\cite{Chen+SICOMP09} called `Divide-and-Color'. It is not based on algebraic tools and finds the witness while solving the decision problem. 
Note that there are many more algorithms/heuristics for $k$-path problem which would be much faster on particular instances. 
A natural heuristic is computing the DFS tree. If the tree has depth at least $k$ the witness is found and otherwise the graph has pathwidth at most $k$. On the other hand, when the pathwidth $p$ is very small (say, $p\le \frac{k}2$), the $(2+\sqrt{2})^pn^{O(1)}$ algorithm of Cygan et al.~\cite{CNK:Hamiltonicity} should be fast. However, in this work we want to focus on algorithms with best guarantees {\em in the worst case}. Disregarding the detailed memory layout of the input graph, all the three algorithms we compare are oblivious to the topology of the graph apart from the parameters $m$ and $k$. In our experiments we use two types of trees with $m=n-1$ as the input graphs. 
The first type (with a unique witness) consists of $\lfloor (k-1)/2\rfloor$-vertex paths joined at a common endvertex; when $k$ is odd two of the paths are extended by an edge, when $k$ is even one path is extended by two edges and one path by one edge. The second type (with $\Omega(n^2)$ witnesses) has $k$ odd and all paths are extended by an edge.

The results can be seen on Fig.~\ref{fig:alg_charts}. We see that both fifo and HKLR are much faster than Divide-and-Color even for very small values of $k$. For 1000-vertex graphs our algorithm fifo finds $(\le 10)$-vertex patterns below 1 second and $(\le 20)$-vertex patterns below 1 hour.
HKLR is considerably slower and the difference is more visible when there are many witnesses. 

\begin{figure}[t]
\begin{minipage}[b]{0.5\linewidth}
\centering
\includegraphics[width=\textwidth]{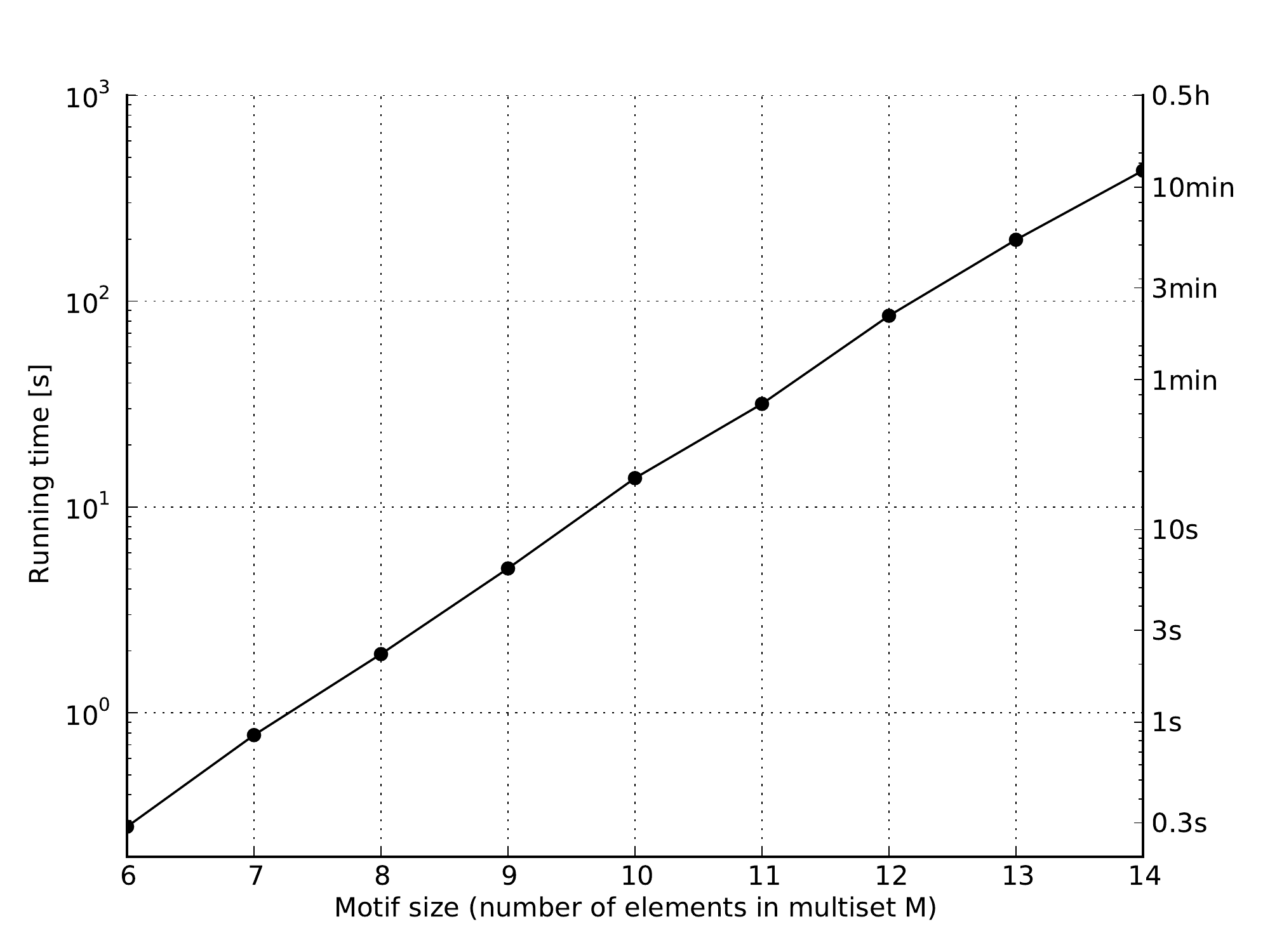}
\end{minipage}
\begin{minipage}[b]{0.5\linewidth}
\centering
\includegraphics[width=\textwidth]{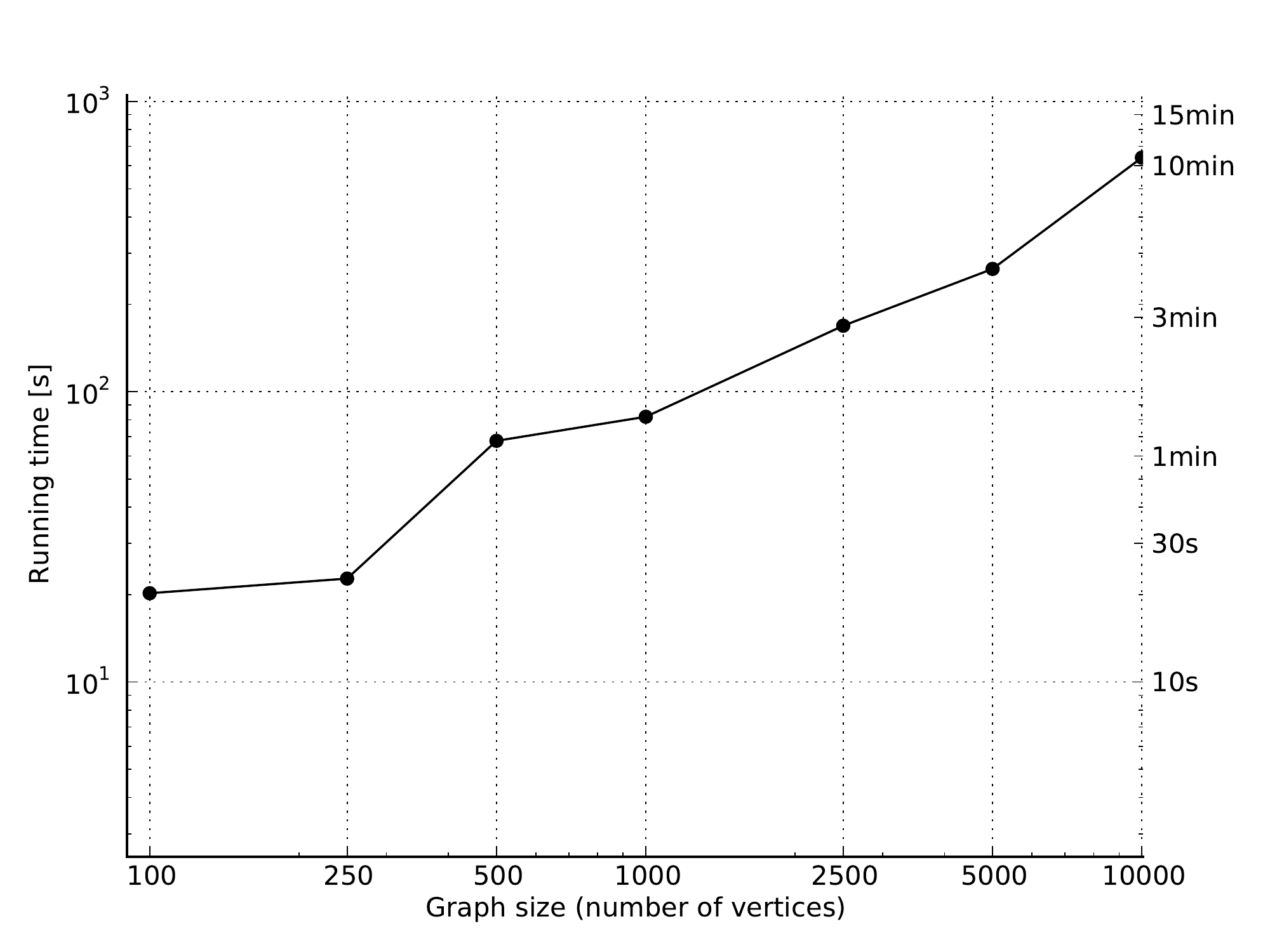}
\end{minipage}
\caption{Running times of witness extraction for Graph Motif problem. 
         Each running time on the graph is the median of 5 runs for the same input instance.
         The left chart: a 8000-vertex 32000-edge graph and $k\in\{6,\ldots,14\}$ 
         (The size of the graph is roughly the same as the PPI network of human).
         The right chart: $k=14$ and the number of vertices varies from 100 to 10000 ($m=4n$).
         (Note that both axes use logarithmic scale.)
         }
\label{fig:motif_charts}
\end{figure}

We have also implemented the $O(2^kk^2\cdot m)$-time decision algorithm \cite{BjorklundKaskiKowalik2013} for the graph motif problem, plugged into the fifo extraction algorithm from the present paper.
In Fig.~\ref{fig:motif_charts} one can see running times of our implementation.
The size of the input instance is typical for the applications in protein-protein interaction networks.
Similarly as in the case of $k$-path, the running time of the decision algorithm is essentially the same regardless of the structure of the graph and the motif, so we just used a random input graph and a random motif.

\section{Extracting a Witness Using a Deterministic Oracle}
\label{sect:deterministic}

The objective of this section is to prove Theorem~\ref{thm:main-deterministic}. Accordingly, we assume we have available a deterministic inclusion oracle. Our strategy is to translate an existing algorithm developed for group testing into the setting of witness extraction (Algorithm~\ref{alg:extract-inclusion} and Lemma~\ref{lem:deterministic-extraction}), and then analyze its performance with respect to the total running time, including the oracle queries (Lemma~\ref{lem:deterministic-extract-runtime}).

Let us first review the setting of classical group testing, and then indicate how to translate classical algorithms to the setting of witness extraction. In group testing, we do not have a family of witnesses, but rather a {\em single} unknown set $D\subseteq U$ consisting of {\em defective} items. Furthermore, instead of an inclusion oracle (that would test whether $D\subseteq Y$ for a query $Y$) we have an {\em intersection oracle} that answers whether $D\cap Y\neq\emptyset$ for a query $Y$. That is, a query tells us whether the query set $Y$ has at least one defective item.

Characteristic to classical group testing algorithms is that they proceed to shrink down the size of the universe $U$ while maintaining the invariant $D\subseteq U$ until $D$ has been identified (that is, $D=U$). Indeed, whenever the (intersection) oracle answers NO, we know that the query $Y$ is disjoint from $D$, and thus can safely delete all points in $Y$ from $U$ without violating the invariant. 

In our setting we have to work with an inclusion oracle and cope with the possibility of the family $\mathcal{F}$ containing more than one witness. Fortunately, it turns out that the setting is not substantially different from group testing. Indeed, in analogy with group testing, we will also proceed to narrow down the universe $U$ but seek to maintain a slightly different invariant, namely ``there exists a $W\in \F$ such that $W\subseteq U$''. In this setting we can narrow down the universe by the following basic procedure: for a subset $A\subseteq U$ we query the inclusion oracle with $Y=U\setminus A$. If the answer is YES, we know that we can safely remove $A$ from $U$ while maintaining the invariant. This basic analogy enables one to transport group testing algorithms into the setting of witness extraction.

In what follows we focus on a translation of one such algorithm, the {\em bisecting algorithm}~\cite{DuHwang1993}. One of its advantages is that it does not need to know the number of defective items in advance, and hence in particular it is suitable for our applications where we want to allow the witnesses to potentially differ in size. 
Moreover, this particular algorithm is convenient in our further modifications for the randomized oracle model (Sect.~\ref{sect:randomized}). 
We give the pseudocode of a ``witness extraction'' version of the bisecting algorithm in pseudocode as Algorithm~\ref{alg:extract-inclusion}. 

{
\begin{algorithm}[t]
\caption{\textsc{ExtractInclusion}$(U)$}
\label{alg:extract-inclusion}
\footnotesize
  Initialize an empty FIFO queue $\mathcal{Q}$\;
  Let $W\leftarrow\emptyset$\;
  Insert $U$ into $\mathcal{Q}$\;
  \While{$\mathcal{Q}$ is not empty}{
    Remove the first set $A$ from $\mathcal{Q}$\;
    \eIf{$|A|=1$}{%
      Let $W\leftarrow W\cup A$\;%
    }{%
      Partition $A$ into $A_1$ and $A_2$ arbitrarily so that $||A_1|-|A_2||\leq 1$\;
      \eIf{\textsc{Includes}$(U\setminus A_1)$\label{line:oracle1}}{%
        Let $U\leftarrow U\setminus A_1$\;
        Insert $A_2$ into $\mathcal{Q}$\;
      }{%
        \eIf{\textsc{Includes}$(U\setminus A_2)$\label{line:oracle2}}{%
          Let $U\leftarrow U\setminus A_2$\;
          Insert $A_1$ into $\mathcal{Q}$\;
        }{%
          Insert both $A_1$ and $A_2$ into $\mathcal{Q}$\;          
        }
      }
    }
  }
  \Return{$W$}
\end{algorithm}
}

The correctness of Algorithm~\ref{alg:extract-inclusion} follows from the fact that our invariant ``there exists a $W\in \F$ such that $W\subseteq U$'' is always satisfied. We remark that Algorithm~\ref{alg:extract-inclusion} has a further minor difference with the original bisection algorithm in that whenever it partitions a set $A$ into $A_1$ and $A_2$ then $A_1$ and $A_2$ are almost of the same size ($||A_1|-|A_2||\leq 1$), whereas the original algorithm $|A_1|=2^{\lceil \log |A| \rceil-1}$ and $|A_2|=|A|-|A_1|$. Du and Hwang~\cite{DuHwang1993} showed that the bisection algorithm performs $O\bigl(k\log\frac{n}{k}\bigr)$ queries. Below we present a self-contained analysis.

\begin{lemma}
\label{lem:deterministic-extraction}
Algorithm~\ref{alg:extract-inclusion} makes at most 
$2k\bigl(\log_2\frac{n}{k}+2\bigr)$ oracle queries.
\end{lemma}

\begin{proof}
We can model the execution of Algorithm~\ref{alg:extract-inclusion} with a 
tree $\mathcal{T}$ whose nodes are the subsets $A$ that have appeared in the queue $\mathcal{Q}$ during execution. A node $A$ is a child of node $B$ if and only if $A$ was obtained by bisecting $B$. In particular $\mathcal{T}$ is a binary tree with at most $k$ leaves and two types of internal nodes: the {\em partition nodes} with two children correspond to splitting a set into two halves, and the {\em cut nodes} with one child correspond to cutting-off a half of a set. Each internal node in $\mathcal{T}$ is associated with 1 or 2 queries. 

Let us order $\mathcal{T}$ arbitrarily so that every partition node has a left child and a right child; let us furthermore call the only child of a cut node the left child. 
For every leaf $v$ form a path $P_v$ up in the tree by first including $v$ into the path and including each subsequent node into $P_v$ as long as we arrived into the node from the left child of the node. Such paths $P_v$ clearly form a partition of nodes in $\mathcal{T}$.

For every cut node $x$, let $D_x$ denote the subset of vertices that was
discarded. For a leaf $v$ let $S_v$ denote the union of all the sets $D_x$ 
on path $P_v$. For any cut nodes $x$ and $y$ on $P_v$, if $x$ is an ancestor 
of $y$ then $|D_x|\geq 2|D_y|-1$.
It follows that there are at most $\lceil\log_2 |S_v| \rceil $ cut nodes 
on $P_v$. Hence the total number of cut nodes is at most 
$\sum_v \lceil\log_2 |S_v|\rceil \le k\bigl(\log_2\frac{n}{k}+1\bigr)$ 
where the sum is over the at most $k$ leaves $v$ in $\mathcal{T}$ and 
the inequality follows from Jensen's inequality (and the fact that 
the sets $S_v$ form a partition of $U\setminus W$, where $W$ is the 
returned witness). Since $\mathcal{T}$ is a binary tree, 
the number of partition nodes is at most $k-1$. Thus there are 
at most $k\bigl(\log_2\frac{n}{k}+2\bigr)$ nodes and 
at most $2k\bigl(\log_2\frac{n}{k}+2\bigr)$ queries.\qed
\end{proof}

A routine information-theoretic argument shows that 
Lemma~\ref{lem:deterministic-extraction} is optimal up to constants,
that is, at least $\log_2\binom{n}{k}\geq k\log_2\frac{n}{k}$ queries 
(bits of information) are needed to identify a unique witness of size $k$ 
in a universe of size $n$. This observation can be strengthened to
the randomized setting via the Yao principle---in expectation at least $\frac{k}{2}\log_2\frac{n}{k}$ queries are required.

We now proceed to analyze Algorithm~\ref{alg:extract-inclusion} with a more natural complexity measure, namely the total time of the extraction procedure, taking into account the time used by the oracle queries. Recall that a function $g:\mathbb{N}\rightarrow\mathbb{N}$ is at least linear if for all $n_1,n_2\in\mathbb{N}$ we have $g(n_1)+g(n_2)\leq g(n_1+n_2)$.

\begin{lemma}
\label{lem:deterministic-extract-runtime}
Suppose the time complexity of the inclusion oracle on a query set of size $n$ 
is $T(n,k)=O(f(k)g(n))$, where $g$ is at least linear.
Then, the running time of Algorithm~\ref{alg:extract-inclusion} 
is $O(k\cdot T(2n,k))$.
\end{lemma}

\begin{proof}
We follow the notation introduced in the proof of Lemma~\ref{lem:deterministic-extraction}. Because there are at most $k-1$ partition nodes, the total time spent at these nodes is $O(k\cdot T(n,k))$. Hence it remains to analyze the time spent at the cut nodes. It suffices to show that for every leaf $v$ of the tree $\mathcal{T}$ the total time spent at the cut nodes in path $P_v$ is $O(T(n,k))$. Observe that at every cut node the size of the universe decreases by a factor of 2. Hence this time is at most $T(n,k)+T(n/2,k)+T(n/4,k)+\ldots+T(1,k)\leq T(2n,k)$ where the last inequality uses the assumption that $g$ is at least linear. \qed
\end{proof}

Lemma~\ref{lem:deterministic-extraction} and 
Lemma~\ref{lem:deterministic-extract-runtime}
now establish Theorem~\ref{thm:main-deterministic}.


\section{Extracting a Witness Using a Randomized Oracle}
\label{sect:randomized}

The objective of this section is to prove Theorem~\ref{thm:main-randomized}. Accordingly, we assume we have available a randomized inclusion oracle that has no false positives but may output a false negative with probability at most $p\leq \frac{1}4$. The outcomes of queries are assumed to be mutually independent as random events.

We start with two simple observations regarding Algorithm~\ref{alg:extract-inclusion} in the context of a randomized oracle. First, since the oracle does not have false positives, the set $W$ output by Algorithm~\ref{alg:extract-inclusion} is always a superset of a witness. Second, by Theorem~\ref{thm:main-deterministic} we know that the algorithm makes at most $Q(n,k)$ queries to extract a witness {\em in the event no false negatives occur in the first $Q(n,k)$ queries}. By the union bound, the probability of this event is at least $1-pQ(n,k)$. This gives us a Monte Carlo algorithm that fails with probability at most $pQ(n,k)$. 

Recalling that we assume we have access to a subroutine that checks whether a given set $W\subseteq U$ satisfies $W\in\mathcal{F}$, we would clearly like to transform the Monte Carlo algorithm into a Las Vegas algorithm that always extracts a witness, and the cost of randomization is only paid in terms of the running time. 

\makeatletter
The Las Vegas algorithm now operates in two stages. Let us call this algorithm Algorithm~\refstepcounter{\algocf@float}\arabic{\algocf@float}\label{alg:second}. In the first stage, we simply run Algorithm~\ref{alg:extract-inclusion} and obtain a set $W$ as output. In the second stage, we insert each element of $W$ into an empty queue $\mathcal{Q}$. Next, as long as $W$ is not a witness, we (1) remove an element $e$ from the head of $\mathcal{Q}$, (2) if {\sc Includes}($W\setminus\{e\})$ returns NO then we insert $e$ at the tail of $\mathcal{Q}$ and otherwise we remove $e$ from $W$. Finally, we return $W$. 
\makeatother

Given that only false negatives may occur, Algorithm~\ref{alg:second} is obviously correct and always returns a witness. It remains to analyze the expected number of queries and the expected running time of Algorithm~\ref{alg:second}.

\begin{lemma}
\label{thm:1se-upper-bound-query-model}
Algorithm~\ref{alg:second} makes in expectation $O(k\log n)$ queries to the randomized inclusion oracle.
\end{lemma}

\begin{proof}
First we bound the expected number of queries in the first stage. 
Recall the tree model of the execution of Algorithm~\ref{alg:extract-inclusion} in the proof of Lemma~\ref{lem:deterministic-extraction}. 
Let us study the model in the presence of false negatives. 
A false negative at line~\ref{line:oracle1} of 
Algorithm~\ref{alg:extract-inclusion} causes the algorithm to view
the set $A_1$ as necessary and continue processing it even if 
it could in be dropped in reality. Similarly, a false negative at 
line~\ref{line:oracle2} causes the algorithm to view $A_2$ as necessary.
In particular, each false negative gets inserted into the queue $\mathcal{Q}$
and hence into the tree $\mathcal{T}$.

Now let us study an arbitrary subtree of $\mathcal{T}$ rooted at a false
negative node. We observe that all such nodes either remain false negative 
nodes, or become exhausted as YES nodes or singleton nodes. (That is, no
node in the subtree is a true negative.) Let us study the process that
creates such a subtree and for convenience ignore the possibility of
singleton nodes exhausting the process. Let $X$ be the random variable 
that tracks the size of the subtree. Because the left and right child 
nodes of each node are independently false negatives with probability $p$, 
we observe that the expectation of $X$ satisfies 
$\mathrm{E}[X]=1+2p\mathrm{E}[X]$. That is, $\mathrm{E}[X]=1/(1-2p)$.
Because $p\leq \frac{1}4$, we have $\mathrm{E}[X]\leq 2$. 
Since each false negative has to interact with true negative and positive
nodes, the expected number of queries in the first stage is, by linearity of
expectation, at most $3Q(n,k)$ by Lemma~\ref{lem:deterministic-extraction}.

Let $W_0$ denote $W$ at the beginning of the second stage. 
For purposes of analysis we divide the second stage into 
two sub-stages. The first sub-stage finishes when $|W|\le 2k$. 
Assume that there was at least one query in the first sub-stage, 
that is, $|W_0|>2k$.
Let $Z$ be the total number of queries in the first sub-stage.
Then $Z=Z_1+Z_2+Z_3$ where $Z_1$ is the number of false negative queries, 
$Z_2$ is the number of positive queries and $Z_3$ is the number of true 
negative queries. First observe that $Z_1$ has the negative binomial 
distribution, that is, $Z_1\sim \mathrm{NB}(|W_0|-2k,p)$, 
and hence $\mathrm{E}[Z_1\;|\;|W_0|]=(|W_0|-2k)\tfrac{p}{1-p}\le |W_0|-2k$. 
It follows that $\mathrm{E}[Z_1] \le \mathrm{E}[|W_0|]-2k \leq 3Q(n,k)$.
Now note that that $Z_2$ is bounded by $|W_0|$, which is bounded by 
the number of queries in the first stage, so $\mathrm{E}[Z_2]\leq 3Q(n,k)$. 
Call an element $e$ of $W$ {\em false} if $W\setminus\{e\}$ contains 
a witness and {\em true} otherwise. Since there are at most 
$k$ true elements, as long as $|W|>2k$ the number of true elements is 
bounded by the number of false elements (if $W$ contains more than 
one witness then all elements of $W$ may be false). 
If $e\in W$ is a true element then the query $W\setminus\{e\}$ 
always returns NO (a true negative); if $e$ is false then 
the query $W\setminus\{e\}$ may return either YES (a true positive)
or NO (a false negative). 
Since elements of $W$ are tested in queue order, $Z_3 \le Z_1+Z_2$ 
and hence $\mathrm{E}[Z_3]\leq 6Q(n,k)$.

Finally consider the second sub-stage, when $|W|\le 2k$. 
Let $t$ be the number of false elements in $W$, $t\le 2k$.
The algorithm iterates through the queue until there is no false element 
in $W$. The number of times we iterate over the whole queue is the maximum 
of $t$ independent random variables, each of geometric distribution 
with success probability $1-p$, which by $p\leq \frac{1}4$ is in expectation 
at most $1+H_t/\ln(1/p)\leq 2H_{2k}\leq 3\ln 2k$ 
(cf.~\cite{Eisenberg2008}). Since in each iteration the algorithm 
performs at most $2k$ queries, the expected number of queries in the second 
sub-stage is then at most $6k\ln 2k$.

The expected number of queries is thus at most $15Q(n,k)+6k\ln 2k$.\qed 
\end{proof}

Theorem~\ref{thm:main-randomized} is now established by
Lemma~\ref{thm:1se-upper-bound-query-model} and the following lemma.

\begin{lemma}
\label{lem:randomized-extract-runtime}
Suppose the time complexity of the randomized inclusion oracle on a query set of size $n$ 
is $T(n,k)=O(f(k)g(n))$, where $g$ is at least linear.
Then, the running time of Algorithm~\ref{alg:second}
is $O(k T(2n,k)+k\log k T(2k,k))$.
\end{lemma}

\begin{proof}
 By Lemma~\ref{lem:deterministic-extract-runtime} the total time of the queries in the first stage that returned a correct answer is bounded by $O(k\cdot T(2n,k))$. 
 
 As argued in the proof of Lemma~\ref{thm:1se-upper-bound-query-model}, all nodes corresponding to false negative queries in {\em both stages} form $O(1)$-sized subtrees of tree $\mathcal{T}$. For every such subtree the parent $p$ of the root of the subtree corresponds to a query with a correct answer. Moreover the size of the instance passed to the oracle in every call in the subtree is bounded by the size of the instance passed to the oracle in the query corresponding to $p$. Hence the expected total time spent at the subtree is asymptotically the same as the time spent at $p$. It follows that all the false negative queries take $O(k\cdot T(2n,k))$ time in expectation. In particular we showed that the first phase takes $O(k\cdot T(2n,k))$ expected time.
 
 For every positive query {\sc Includes}($W\setminus\{e\})$ in the second stage there is the corresponding (false negative) leaf corresponding to the singleton $\{e\}$ in tree $\mathcal{T}$. Hence the total time of positive queries in the second stage is bounded by the time of the first phase.
 
 Now we focus on true negative queries in the first sub-stage of the second stage. Consider a single pass of the algorithm through all the elements in the queue.
 Within this pass there are at most $k$ true negatives (since witnesses are of size at most $k$). Moreover, since in the second phase there are at least $2k$ elements in the queue, we can injectively assign to each of the true negative queries a false negative or a positive query for an instance of the same or larger size. Hence, the total time of true negative queries in the first sub-stage of the second stage is bounded by the total time of 
  false negative and positive queries which we already bounded by $O(k\cdot T(2n,k))$.
 
 We are left with bounding the time of true negatives in the second sub-stage of the second stage. However, since then $|W|\le 2k$, each query takes just $T(2k,k)$ time. In the proof of Lemma~\ref{thm:1se-upper-bound-query-model} we showed that the total number of queries in the second sub-stage is $O(k\log k)$, so the desired bound follows. \qed
 \end{proof}


\section{Implementation of Finite Field Arithmetic}
\label{sec:gf-impl}

The most critical subroutines of the $k$-path inclusion oracle we implemented are operations of addition and multiplication in a finite field $\field{2^q}$. The choice of $q$ is important: the oracle returns a false negative with probability at most $\frac{2k-1}{2^q}$.
We can assume that $k\le 30$, for otherwise the oracle runs too long. It follows that to guarantee low error probability, say, at most $2^{-20}$, 
it suffices to pick $q=26$. 

\begin{figure}[t]
\begin{minipage}[b]{0.5\linewidth}
\centering
\includegraphics[width=\textwidth]{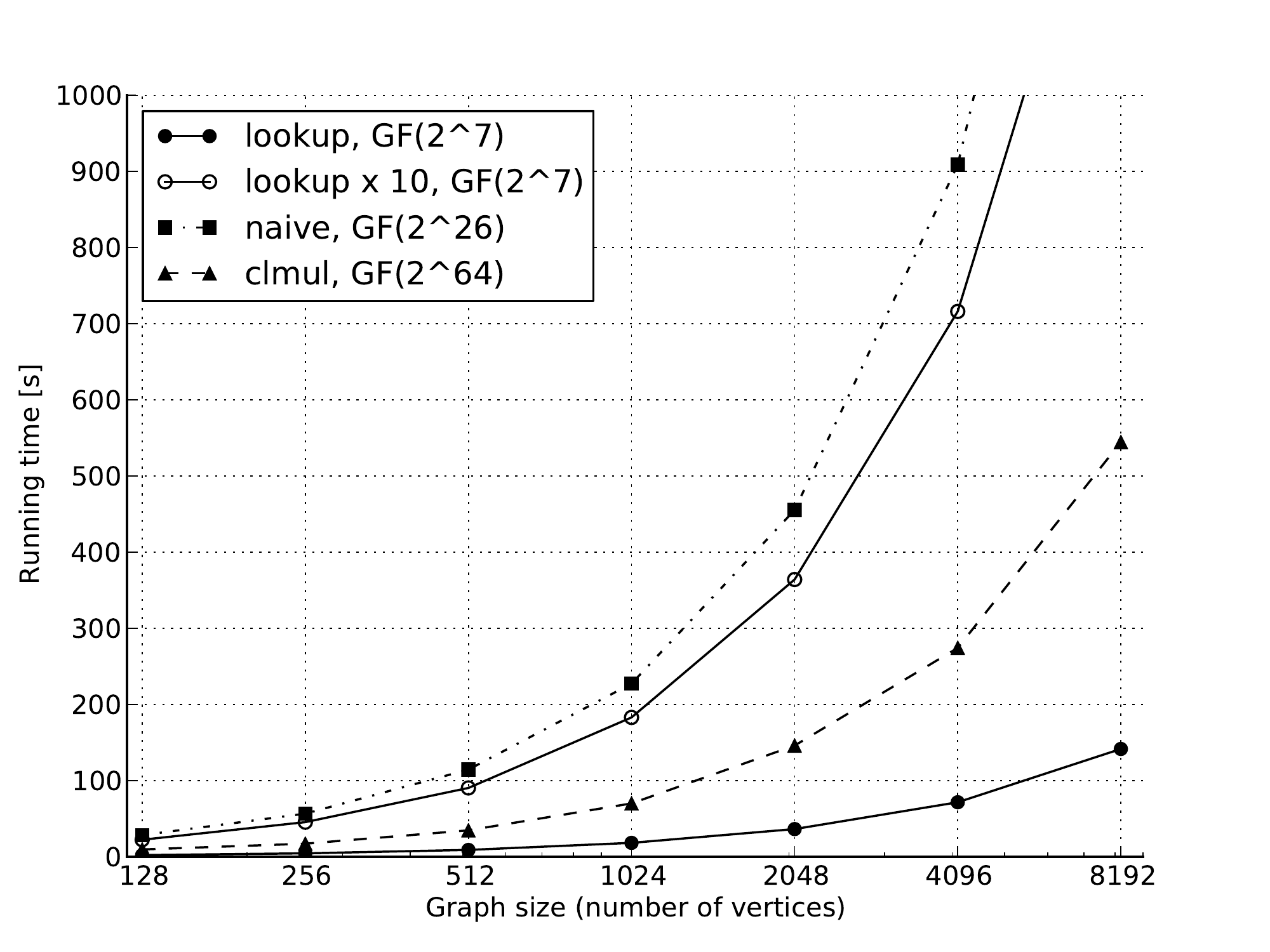}

{\small $k=16$, decision algorithm}
\end{minipage}
\begin{minipage}[b]{0.5\linewidth}
\centering
\includegraphics[width=\textwidth]{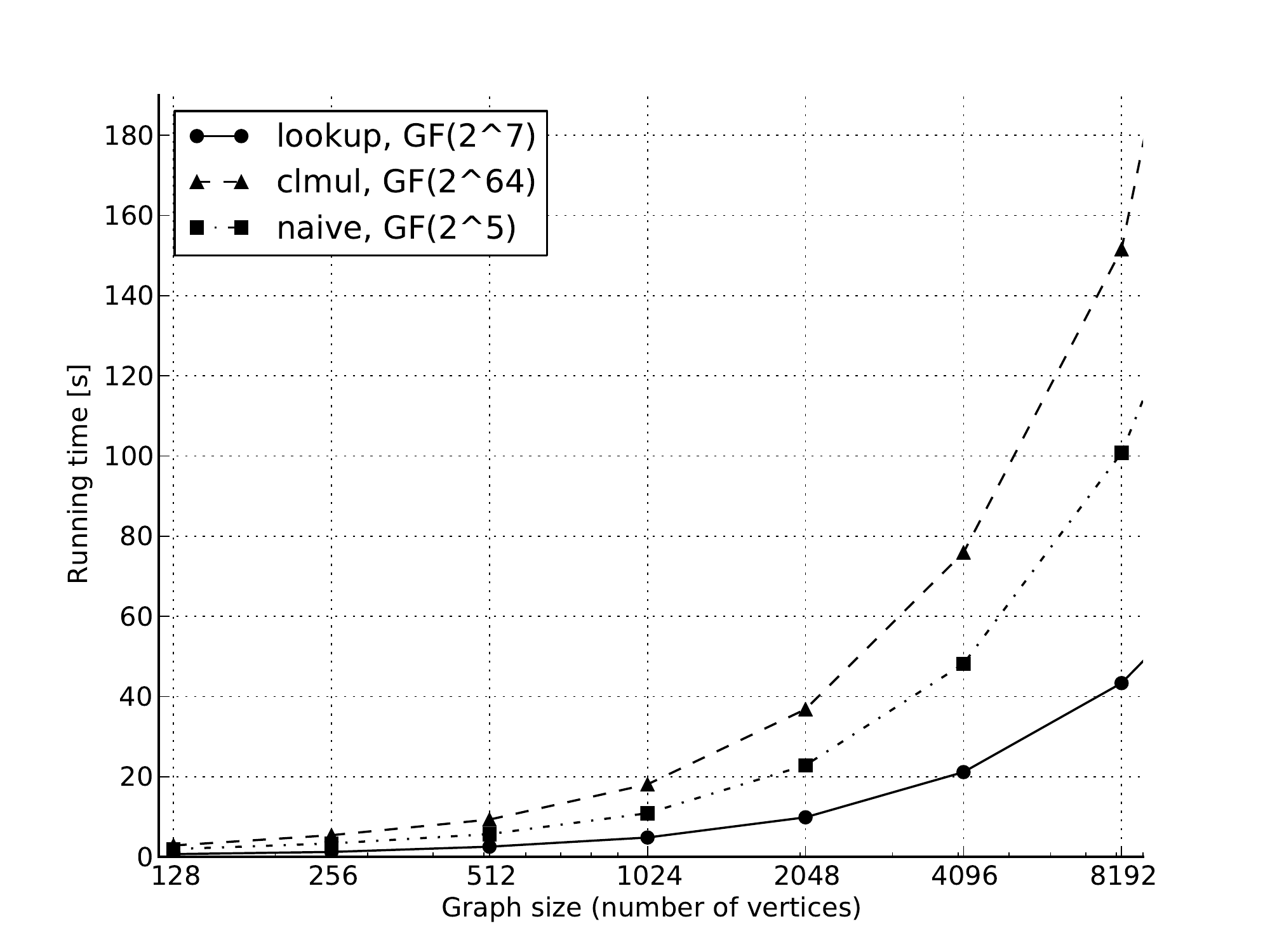}

{\small $k=12$, witness extraction}
\end{minipage}
\caption{\label{fig:gf-impl}Comparison of three implementations of $\field{2^q}$ arithmetic.
Left: (single run of) $k$-path decision oracle for instances with {\em no} solution. The pattern size is fixed as $k=16$ and the number of vertices $n$ varies.
Right: fifo algorithm using $k$-path decision oracle for instances with exactly one solution (each running time on the graph is the median of 5 runs for the same input instance). 
The pattern size is fixed as $k=12$ and the number of vertices $n$ varies.
Running times on a 2.53-GHz Intel Xeon CPU.
}
\end{figure}

Let us recall that elements of $\field{2^q}$ correspond to polynomials of degree at most $q-1$ with coefficients from $\field{2}$.
Such a polynomial is conveniently represented as a $q$-bit binary number.
The addition in $\field{2^q}$ corresponds to addition of two polynomials, that is, the symmetric difference (xor) of the binary representations. Multiplication is performed by (a) multiplying the polynomials and (b) returning the remainder of the division of the result by a primitive degree-$q$ polynomial; this is easily implemented in $O(q)$ word operations. We refer to this implementation as `naive'.

One can observe that step (a) above corresponds to carry-less multiplication of two binary numbers, that is, the usual multiplication without generating carries ($011\times 011=101$). Such multiplication of two 64-bit numbers is available as a single instruction (PCLMULQDQ) on a number of modern Intel and AMD architectures. Using the fact that there is an only 5-term primitive polynomial of degree 64, step (b) can be implemented using bit shifts and xors~\cite{gueron2010efficient}. We refer to this implementation as `clmul'.

The third natural option is to {\em precompute} the whole multiplication table (using the naive algorithm) before running the oracle.
This takes $4^q\lceil q/8\rceil$ bytes of memory, so can be considered only for small values of $q$, say $q\le 12$ (even for $q=12$ the precomputation time is negligible at substantially less than a second). We refer to this implementation as `lookup'.

The left chart of Fig.~\ref{fig:gf-impl} shows the comparison of the three implementations of $\field{2^q}$ arithmetic used in a {\em single run} of the decision oracle. For `naive' we use $q=26$ and for `clmul' $q=64$. For `lookup' we use $q=7$ because for smaller values of $q$ the running time is roughly the same; nevertheless since in the tests we look for a pattern of size 16, it gives just a bound of $\frac{1}4$ for error probability. To squeeze the probability down to $2^{-20}$ one can run the oracle 10 times and return the conjunction of the results. We see that although `lookup' is faster than `clmul' when the oracle is called once, it is much slower when we repeat the oracle call 10 times (note also that clmul provides error probability $2^{-59}$). The `naive' method is worse than the other two. 

 \begin{figure}[!ht]
  \includegraphics[width=\textwidth]{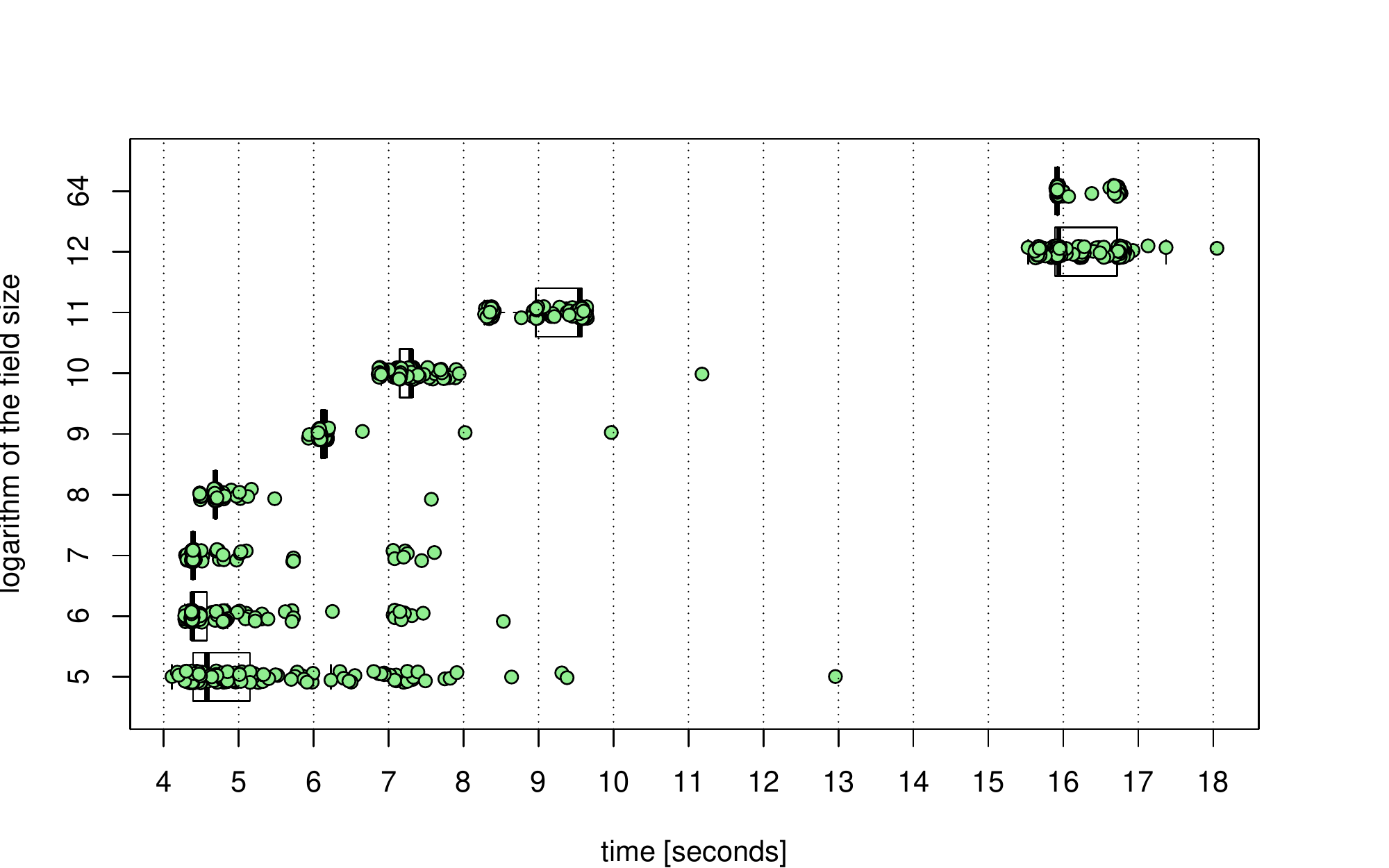}
  \label{fig:boxplots}
  \caption{Statistics for 200 runs of the extraction algorithm ($n=1000$, $k=12$) using lookup implementation for $q=5,\ldots,12$ and clmul implementation ($q=64$). Running time of each execution is visualized as a green circle. All experiments for a given field size $2^q$ are summarized using a boxplot showing a first and third quantile and the mean (thick vertical line).}
 \end{figure}

 Note however that, if we aim at {\em finding} a witness, by Theorem~\ref{thm:main-randomized} it suffices to guarantee that error probability is at most $\frac{1}4$, hence for $k\le 16$ we can pick $q=7$. The advantage of our witness extraction algorithm fifo is that even if it gets a false answer from the oracle, it will discover the mistake in the future. Indeed, the right chart of  Fig.~\ref{fig:gf-impl} shows that using $\field{2^7}$ with `lookup' outperforms using $\field{2^{64}}$ with `clmul', roughly by a factor of four. 
The value $q=7$ here is carefully chosen. One one hand, we want $q$ to be large to get small error probability for a single query and thus small variance of the whole extraction running time. On the other hand, at our machine the multiplication table for $q=8$ does not fit into L1 cache (of size 32K) what results in increase in the median running time. In the table below (see also Fig.~\ref{fig:boxplots}) we show statistics for 200 runs of the extraction algorithm ($n=1000$, $k=12$) using `lookup' for $q=5,6,\ldots,12$ and `clmul' ($q=64$). 

{\footnotesize
\begin{center}
\begin{tabular}{|c|r|r|r|r|r|r|r|r|r|}\hline
logarithm of the field size &  5 & 6 & 7 & 8 & 9 & 10 & 11 & 12 & 64 (clmul)\\\hline
median [sec] &  4.58 & 4.38 & 4.39 & 4.69 & 6.15 & 7.30 & 9.55 & 15.94 & 15.92\\\hline
maximum [sec] & 12.96 &  8.53 & 7.61 & 7.57 & 9.97 & 11.18 & 9.65 & 18.05 & 16.77\\\hline
standard deviation [sec] & 1.16 & 0.68 & 0.61 & 0.22 & 0.30 & 0.34 & 0.50 & 0.43 & 0.25 \\\hline
\end{tabular}
\end{center}
}
 
 Clearly, for $q=8,9,\ldots,12$ we get increased number of cache misses (the 256K L2 cache could fit the table only for $q\le 8$). 
 The running times are concentrated very well around the median for $q=7,8,9$ (on the picture the first and third quantile got merged with the median).
 For $q=10,11,12$ we observe increasing variance. This is caused by the fact that the arithmetic operations are performed on random numbers,
 thus the number of cache misses becomes a random variable (and its expectation increases with $q$). 
 For $q=64$ and clmul implementation we get excellent concentration again because the error probability is very small (most likely there was no single error during the 200 runs) and in this method there are no cache misses.


\bibliographystyle{abbrv}

\newpage

\appendix

 \section{$\field{2^{64}}$-multiplication using PCLMULQDQ instruction}
 \label{sec:clmul-code}
 
 Below we present the code for multiplication in $\field{2^{64}}$ using a single PCLMULQDQ instruction, 6 bitwise shifts and 7 bitwise xors.
 The algorithm is based on the work of Gueron and Kounavis~\cite{gueron2010efficient}.
 
 {\footnotesize
 \begin{verbatim}
uint64_t gf2q_mul (uint64_t x, uint64_t y)
{
  uint64_t xy [2];
  xy [0] = x;
  xy [1] = y;

  __m128i ab = _mm_loadu_si128((__m128i*) xy);

  uint64_t X[2], tmp2, tmp3, tmp4;
  __m128i tmp1;
  
  tmp1 = _mm_clmulepi64_si128(ab, ab, 0x01);
  _mm_storeu_si128((__m128i*)X,tmp1);
  tmp2 = X[1];
  tmp3 = tmp2 >> 63;
  tmp4 = tmp2 >> 61;
  tmp3 = tmp3 ^ tmp4;
  tmp4 = tmp2 >> 60;
  tmp3 = tmp3 ^ tmp4;
  tmp2 = tmp3 ^ tmp2; 
  tmp4 = tmp2 << 1;
  tmp3 = tmp2 ^ tmp4;
  tmp4 = tmp2 << 3;
  tmp3 = tmp3 ^ tmp4;
  tmp4 = tmp2 << 4;
  tmp3 = tmp3 ^ tmp4;
  return tmp3 ^ X[0];
}
\end{verbatim}
}

\section{Running time data}

In this section we include exact running times used to generate the charts in the main part of the paper.

Running times of various algorithms for a graph with exactly one witness (data for Figure~\ref{fig:alg_charts}, upper left).
Each running time is the median of 5 runs for the same input instance.
The input graph has $1000$ vertices and $k\in\{6,\ldots,18\}$.

{\scriptsize
\begin{center}
\begin{tabular}{|c|r@{\,}|r@{\,}|r@{\,}|r@{\,}|r@{\,}|r@{\,}|r@{\,}|r@{\,}|r@{\,}|r@{\,}|r@{\,}|r@{\,}|r@{\,}|r|}\hline
\multicolumn{1}{|c|}{$k$} & \multicolumn{1}{|c|}{6}  & \multicolumn{1}{|c|}{7}  & \multicolumn{1}{|c|}{8}  & \multicolumn{1}{|c|}{9}  & \multicolumn{1}{|c|}{10}  & \multicolumn{1}{|c|}{11}  & \multicolumn{1}{|c|}{12}  & \multicolumn{1}{|c|}{13}  & \multicolumn{1}{|c|}{14}  & \multicolumn{1}{|c|}{15}  & \multicolumn{1}{|c|}{16}  & \multicolumn{1}{|c|}{17}  & \multicolumn{1}{|c|}{18} \\ \hline
fifo & 0.05 & 0.15 & 0.3 & 0.77 & 1.65 & 3.37 & 8.08 & 15.69 & 38.54 & 82.44 & 113.2 & 227.11 & 610.27\\
HKLR & 0.16 & 0.51 & 1.04 & 3.04 & 8.94 & 16.22 & 36.44 & 87.22 & 158.37 & 386.4 & 770.37 & 1585.88 & 3866.21\\
D \& C & 1.52 & 7.11 & 22.43 & 247.5 & 711.34 & 3164.32 & 12015.78 & & & & & & \\\hline
\end{tabular}
\end{center}
}

Running times of various algorithms for a graph with exactly one witness (data for Figure~\ref{fig:alg_charts}, upper right).
Each running time is the median of 5 runs for the same input instance.
The size of the pattern is $k=14$ and the number of vertices $n$ varies.

{\scriptsize
\begin{center}
\begin{tabular}{|c|r|r|r|r|r|r|r|}\hline
\multicolumn{1}{|c|}{$n$} & \multicolumn{1}{|c|}{100}  & \multicolumn{1}{|c|}{250}  & \multicolumn{1}{|c|}{500}  & \multicolumn{1}{|c|}{1000}  & \multicolumn{1}{|c|}{2500}  & \multicolumn{1}{|c|}{5000}  & \multicolumn{1}{|c|}{10000} \\\hline
fifo & 2.93 & 6.12 & 19.22 & 39.43 & 68.9 & 113.19 & 239.13\\
HKLR & 19.66 & 68.73 & 78.84 & 179.64 & 348.33 & 625.68 & 1326.63\\\hline
\end{tabular}
\end{center}
}

Running times of various algorithms for a graph with $\Omega(n^2)$ witnesses (data for Figure~\ref{fig:alg_charts}, lower left).
Each running time is the median of 5 runs for the same input instance.
The input graph has $1000$ vertices and $k\in\{7,9,\ldots,17\}$.

{\scriptsize
\begin{center}
\begin{tabular}{|c|r|r|r|r|r|r|r|}\hline
\multicolumn{1}{|c|}{$k$} & \multicolumn{1}{|c|}{7}  & \multicolumn{1}{|c|}{9}  & \multicolumn{1}{|c|}{11}  & \multicolumn{1}{|c|}{13}  & \multicolumn{1}{|c|}{15}  & \multicolumn{1}{|c|}{17} \\\hline
fifo & 0.01 & 0.08 & 0.38 & 1.86 & 7.93 & 38.55\\
HKLR & 0.16 & 1.05 & 6.64 & 46.09 & 228.02 & \\
D \& C & 4.61 & 189.48 & 3258.19& & &\\\hline
\end{tabular}
\end{center}
}

Running times of various algorithms for a graph with $\Omega(n^2)$ witnesses (data for Figure~\ref{fig:alg_charts}, lower right).
Each running time is the median of 5 runs for the same input instance.
The size of the pattern is $k=15$ and the number of vertices $n$ varies.

{\scriptsize
\begin{center}
\begin{tabular}{|c|r|r|r|r|r|r|r|}\hline
\multicolumn{1}{|c|}{$n$} & \multicolumn{1}{|c|}{100}  & \multicolumn{1}{|c|}{250}  & \multicolumn{1}{|c|}{500}  & \multicolumn{1}{|c|}{1000}  & \multicolumn{1}{|c|}{2500}  & \multicolumn{1}{|c|}{5000}  & \multicolumn{1}{|c|}{10000} \\\hline
fifo & 1.04 & 1.88 & 3.92 & 7.94 & 21.02 & 39.35 & 78.13\\
HKLR & 33.68 & 64.77 & 135.93 & 235.23 & 612.73 & 1771.58 &\\\hline
\end{tabular}
\end{center}
}

Experiments for the graph motif problem  (data for Figure~\ref{fig:motif_charts}, left).
Each running time is the median of 5 runs for the same input instance.
The input graph has $8000$ vertices, $32000$ edges and $k\in\{6,\ldots,14\}$.

{\footnotesize
\begin{center}
\begin{tabular}{|c|r|r|r|r|r|r|r|r|r|}\hline
\multicolumn{1}{|c|}{$k$} & \multicolumn{1}{|c|}{6}  & \multicolumn{1}{|c|}{7}  & \multicolumn{1}{|c|}{8}  & \multicolumn{1}{|c|}{9}  & \multicolumn{1}{|c|}{10}  & \multicolumn{1}{|c|}{11}  & \multicolumn{1}{|c|}{12}  & \multicolumn{1}{|c|}{13}  & \multicolumn{1}{|c|}{14} \\\hline
time (s) & 0.28 & 0.78 & 1.93 & 5.03 & 13.81 & 31.69 & 84.88 & 198.55 & 430.44\\\hline
\end{tabular}
\end{center}
}

Experiments for the graph motif problem  (data for Figure~\ref{fig:motif_charts}, right).
Each running time is the median of 5 runs for the same input instance.
The size of the motif is fixed as $k=14$ and the number of vertices $n$ varies. The number of edges is always $m=4n$.

{\footnotesize
\begin{center}
\begin{tabular}{|c|r|r|r|r|r|r|r|r|r|}\hline
\multicolumn{1}{|c|}{$n$} & \multicolumn{1}{|c|}{100}  & \multicolumn{1}{|c|}{250}  & \multicolumn{1}{|c|}{500}  & \multicolumn{1}{|c|}{1000}  & \multicolumn{1}{|c|}{2500}  & \multicolumn{1}{|c|}{5000}  & \multicolumn{1}{|c|}{10000} \\\hline
time (s) & 20.22 & 22.71 & 67.71 & 81.96 & 168.73 & 264.78 & 640.06\\\hline
\end{tabular}
\end{center}
}

Comparison of three implementations of $\field{2^q}$ arithmetic (data for Figure~\ref{fig:gf-impl}, left).
A single run of $k$-path decision oracle for instances with {\em no} solution. 
The pattern size is fixed as $k=16$ and the number of vertices $n$ varies.

{\scriptsize
\begin{center}
\begin{tabular}{|c|r|r|r|r|r|r|r|}\hline
\multicolumn{1}{|c|}{$k$} & \multicolumn{1}{|c|}{128}  & \multicolumn{1}{|c|}{256}  & \multicolumn{1}{|c|}{512}  & \multicolumn{1}{|c|}{1024}  & \multicolumn{1}{|c|}{2048}  & \multicolumn{1}{|c|}{4096}  & \multicolumn{1}{|c|}{8192} \\\hline
lookup, $\field{2^7}$ & 2.24 & 4.55 & 9.04 & 18.3 & 36.41 & 71.61 & 141.61\\
lookup $\times 10$, $\field{2^7}$ & 22.4 & 45.5 & 90.4 & 183.0 & 364.1 & 716.1 & 1416.1\\
naive, $\field{2^{26}}$ & 28.42 & 56.65 & 114.48 & 227.37 & 455.36 & 908.72 & 1817.9\\
clmul, $\field{2^{64}}$ & 9.64 & 17.2 & 34.73 & 70.08 & 146.06 & 274.27 & 544.51\\\hline
\end{tabular}
\end{center}
}

Comparison of three implementations of $\field{2^q}$ arithmetic (data for Figure~\ref{fig:gf-impl}, right).
Fifo algorithm using $k$-path decision oracle for instances with exactly one solution (each running time on the graph is the median of 5 runs for the same input instance). 
The pattern size is fixed as $k=12$ and the number of vertices $n$ varies.

{\scriptsize
\begin{center}
\begin{tabular}{|c|r|r|r|r|r|r|r|r|}\hline
\multicolumn{1}{|c|}{$k$} & \multicolumn{1}{|c|}{128}  & \multicolumn{1}{|c|}{256}  & \multicolumn{1}{|c|}{512}  & \multicolumn{1}{|c|}{1024}  & \multicolumn{1}{|c|}{2048}  & \multicolumn{1}{|c|}{4096}  & \multicolumn{1}{|c|}{8192}  & \multicolumn{1}{|c|}{16384} \\\hline
clmul, $\field{2^{64}}$ & 2.82 & 5.39 & 9.29 & 18.08 & 36.81 & 75.91 & 151.57 & 343.96\\
naive, $\field{2^{5}}$ & 1.93 & 3.32 & 5.66 & 10.87 & 22.83 & 48.16 & 100.77 & 197.77\\
lookup, $\field{2^7}$ & 0.67 & 1.21 & 2.52 & 4.81 & 9.85 & 21.13 & 43.33 & 83.54\\\hline
\end{tabular}
\end{center}
}

\end{document}